\documentclass[11pt,a4paper]{article}
\usepackage[T1]{fontenc}
\usepackage[utf8]{inputenc}
\usepackage{lmodern}
\usepackage[english]{babel}
\usepackage{amsmath,amssymb,amsthm,mathtools}
\usepackage{geometry}
\usepackage{authblk}
\usepackage[colorlinks=true,linkcolor=blue,citecolor=blue,urlcolor=blue]{hyperref}
\numberwithin{equation}{section}
\newtheorem{theorem}{Theorem}[section]
\newtheorem{lemma}[theorem]{Lemma}

\theoremstyle{remark}

\newcommand{\K}{\mathbb K}

\newcommand{\Mat}{\operatorname{Mat}}
\newcommand{\qdet}{\operatorname{det}_{q}}
\newcommand{\one}{\mathbf1}

\title{Commuting Families of Flag Quasideterminants\\
in Reflection Equation Algebras of Type $A$}
\author[1,2]{Dmitry V. Talalaev\thanks{E-mail: dtalalaev@yandex.ru}}
\affil[1]{\small Lomonosov Moscow State University, 119991 Moscow, Russia}
\affil[2]{\small P. G. Demidov Yaroslavl State University,\newline 14 Sovetskaya St., Yaroslavl 150003, Russia}
\date{October 2026}

\begin{document}
\maketitle
\begin{abstract}
We consider a family of northeastern antidiagonal quasideterminants
associated with the successive blocks of the generating matrix of a
reflection equation (RE) algebra of type $A_{n-1}$. Our main result
proves the conjecture on their pairwise commutativity formulated in
\cite{Talalaev2024}. The proof uses the triangular realization of the
RE algebra in a quantum enveloping algebra and the flag property of
quasideterminants. We also discuss the interpretation of these elements
in terms of quantum cluster coordinates and their possible connection
with generalized Toda systems.
\end{abstract}


\section{Introduction}

The construction of explicitly defined commutative subalgebras in
quantum matrix algebras and quantum universal enveloping algebras is
one of the basic algebraic problems in the theory of integrable
systems. Standard examples include Bethe subalgebras arising from
transfer matrices and Gelfand--Tsetlin subalgebras generated by the
centers of chains of nested algebras.

In \cite{Talalaev2024}, quasideterminants associated with consecutive
northeastern blocks of the generating matrix of an RE algebra were
introduced. Their commutativity was established in low ranks, while
the general statement was formulated as a conjecture. Here we give a
proof based on a triangular realization of the quantum group and an
embedding of the RE algebra into it.

The crucial point is that the quasideterminants under consideration
are noncommutative flag (Pl\"ucker-type) coordinates and are invariant
under multiplication by unipotent triangular matrices over a skew
field. Applying this property to the RLL realization of the RE matrix,
we express the relevant elements through ratios of nested quantum RTT
minors and Cartan factors. The centrality of a quantum determinant in
a square RTT algebra implies commutativity of the nested minors; the
remaining weight factors cancel because the corresponding exponent
matrix is symmetric.

We then explain the analogy with the Gelfand--Tsetlin construction and
discuss a connection with a distinguished family of quantum cluster
coordinates. We also compare our construction with the $k$-chop
reduction developed by Skrypnyk \cite{Skrypnyk2000}. The latter takes
place in $U(\mathfrak{gl}_n)$ and yields a larger spectral family.

\section{The RE algebra and the main theorem}

We begin with the basic definitions. Let $\K=\mathbb Q(q)$ and
$P=\sum_{i,j}E_{ij}\otimes E_{ji}$. The standard FRT matrix of
type $A_{n-1}$ is
\begin{equation}\label{eq:RF}
 R_F=q\sum_iE_{ii}\otimes E_{ii}
 +\sum_{i\ne j}E_{ii}\otimes E_{jj}
 +(q-q^{-1})\sum_{i<j}E_{ij}\otimes E_{ji}.
\end{equation}
We also use the braiding $\check R_{\mathrm T}=R_FP$. Denote by
$\mathcal A_q(n)$ the RE algebra generated by the entries of the matrix
$L=(l_{ij})$ subject to
\begin{equation}\label{eq:RE}
 \check R_{\mathrm T}L_1\check R_{\mathrm T}L_1
 =L_1\check R_{\mathrm T}L_1\check R_{\mathrm T},
 \qquad L_1=L\otimes\one.
\end{equation}

For $1\leq k\leq n$ define the submatrix and its quasideterminant by
\begin{equation}\label{eq:Ik}
 B_k=L_{\{1,\ldots,k\},\{n-k+1,\ldots,n\}},
 \qquad I_k=|B_k|_{k,1}.
\end{equation}
The indices $k,1$ specify the row and column of the distinguished
entry $l_{k,n-k+1}$. Recall the definition
\begin{equation}\label{eq:qd}
 |A|_{ij}=a_{ij}-r_i^{\widehat j}
 (A^{ij})^{-1}c_j^{\widehat i},
\end{equation}
where $A^{ij}$ is obtained from $A$ by deleting row $i$ and column
$j$. All expressions below are interpreted in a suitable skew field of
fractions in which the required inverses exist.

\begin{theorem}\label{thm:main}
For generic $q$, the northeastern flag quasideterminants
$I_1,\ldots,I_n$ commute pairwise:
\begin{equation}\label{eq:main}
 \boxed{[I_i,I_j]=0,\qquad 1\leq i,j\leq n.}
\end{equation}
\end{theorem}

\section{Triangular covariance}

Let $D$ be a skew field and $J$ the antidiagonal permutation matrix,
$J_{ij}=\delta_{i,n+1-j}$. For $A\in\Mat_n(D)$ set
\begin{equation}\label{eq:Fk}
 F_k(A)=\left|
 A_{\{1,\ldots,k\},\{n-k+1,\ldots,n\}}
 \right|_{k,1}
 =|(AJ)_{[1,k],[1,k]}|_{k,k}.
\end{equation}

\begin{lemma}[Triangular covariance]\label{lem:triang}
If $B,C$ are invertible lower triangular matrices over $D$, then
\begin{equation}\label{eq:cov}
 F_k(BAC)=b_{kk}\,F_k(A)\,c_{n-k+1,n-k+1},
\end{equation}
provided that all the quasideterminants involved are defined. In
particular, $F_k(BAC)=F_k(A)$ when $B$ and $C$ are lower unitriangular.
\end{lemma}
\begin{proof}
The assertion follows from the elementary properties of
quasideterminants established by Gelfand and Retakh
\cite[\S~I.2.5]{GelfandRetakh1997}. Its invariant part is a version
of the corresponding property of noncommutative flag coordinates
\cite[\S~II.2.7]{GelfandRetakh1997}. For completeness, we recall a
short proof that keeps track of the order of noncommuting factors.

Put $X=AJ$. Then $(BAC)J=BX(JCJ)$ and $JCJ$ is upper triangular.
The leading block of size $k$ in this product is
$B_{[k]}X_{[k]}(JCJ)_{[k]}$. Its last principal quasideterminant is
the last diagonal factor in its Gauss decomposition over $D$.
Left multiplication by a lower triangular matrix multiplies this
factor on the left by $b_{kk}$; right multiplication by an upper
triangular matrix multiplies it on the right by
$(JCJ)_{kk}=c_{n-k+1,n-k+1}$. The displayed order of the factors is
preserved, proving \eqref{eq:cov}.
\end{proof}

\section{The triangular RLL realization}

Another essential ingredient is the triangular RLL realization of
the quantum group $U_q(\mathfrak{gl}_n)$ \cite{FRT1990}. It uses an
upper triangular matrix $T^+$ and a lower triangular matrix $T^-$
satisfying
\begin{equation}\label{eq:RLL}
 R_FT^\pm_1T^\pm_2=T^\pm_2T^\pm_1R_F,
 \qquad R_FT^+_1T^-_2=T^-_2T^+_1R_F.
\end{equation}
Their diagonal entries obey
\begin{equation}\label{eq:K}
 K_s=t^+_{ss}=(t^-_{ss})^{-1},\qquad [K_s,K_t]=0.
\end{equation}

The matrix $M=T^+(T^-)^{-1}$ satisfies the reflection equation with
the braiding $PR_F$. Set $J_{ij}=\delta_{i,n+1-j}$. Directly from
\eqref{eq:RF},
\begin{equation}\label{eq:conventions}
 (J\otimes J)(PR_F)(J\otimes J)=R_FP.
\end{equation}
Hence we obtain a homomorphism
\begin{equation}\label{eq:embedding}
 \Phi:\mathcal A_q(n)\longrightarrow U_q(\mathfrak{gl}_n),
 \qquad \Phi(L)=JT^+(T^-)^{-1}J.
\end{equation}
For generic $q$, this is an injective realization of the RE algebra;
see the Joseph--Letzter--Rosso theory and
\cite[\S~1.4]{JordanWhite2020}.

Put
\begin{equation}\label{eq:BC}
 B=JT^+J,\qquad C=J(T^-)^{-1}J.
\end{equation}
Here $B$ is lower triangular and $C$ is upper triangular. Lemma
\ref{lem:triang} gives
\begin{equation}\label{eq:IQ}
 \boxed{\Phi(I_k)=K_{n-k+1}Q_k,\qquad
 Q_k=\left|C_{\{1,\ldots,k\},\{n-k+1,\ldots,n\}}
 \right|_{k,1}.}
\end{equation}

\begin{lemma}\label{lem:C}
The matrix $C$ satisfies the RTT relation
\begin{equation}\label{eq:CRTT}
 R_FC_1C_2=C_2C_1R_F.
\end{equation}
\end{lemma}
\begin{proof}
Write $V=(T^-)^{-1}$. Inverting the relation
$R_FT^-_1T^-_2=T^-_2T^-_1R_F$ yields
$R_F^{-1}V_1V_2=V_2V_1R_F^{-1}$.
As $(J\otimes J)R_F(J\otimes J)=(R_F)_{21}$,
$C=JVJ$ satisfies the analogous relation with
$(R_F)_{21}^{-1}$. Finally,
$(R_F)_{21}^{-1}=R_F-(q-q^{-1})P$, while
$PC_1C_2=C_2C_1P$; thus the two RTT forms are equivalent.
\end{proof}

\section{Nested quantum minors and proof of the theorem}

For a square RTT block $A=(a_{ij})$ of size $k$, we use the
quantum determinant
\begin{equation}\label{eq:qdet}
 \qdet(A)=\sum_{\sigma\in S_k}
 (-q^{-1})^{\ell(\sigma)}
 a_{1,\sigma(1)}\cdots a_{k,\sigma(k)}.
\end{equation}
Define the nested northeastern quantum minors by
\begin{equation}\label{eq:Delta}
 \Delta_0=1,\qquad
 \Delta_k=\qdet
 C_{\{1,\ldots,k\},\{n-k+1,\ldots,n\}}.
\end{equation}
Let $\mathcal B_k$ denote the subalgebra of
$U_q(\mathfrak{gl}_n)$ generated by all entries of the corresponding
square block.

\begin{lemma}[Centers of nested RTT subalgebras]\label{lem:centers}
There are inclusions
\[
 \mathcal B_1\subset\mathcal B_2\subset\cdots\subset
 \mathcal B_n,
 \qquad \Delta_k\in Z(\mathcal B_k).
\]
In particular,
\begin{equation}\label{eq:commDelta}
 [\Delta_i,\Delta_j]=0\quad\text{for all }i,j.
\end{equation}
\end{lemma}
\begin{proof}
By \eqref{eq:CRTT}, each of the selected square blocks is itself
an RTT matrix. The quantum determinant of a square RTT matrix is
central in the algebra generated by its matrix entries
\cite[\S~3.4]{Isaev2004}; see also \cite{EtingofRetakh1998}.
For $i<j$, the smaller northeastern block is contained in the
larger one, so $\Delta_i\in\mathcal B_j$ while
$\Delta_j\in Z(\mathcal B_j)$. This proves \eqref{eq:commDelta}.
\end{proof}

\begin{lemma}[Flag quasideterminants as ratios of minors]
\label{lem:ratio}
In the skew field of fractions,
\begin{equation}\label{eq:ratio}
 \boxed{Q_k=(-q^{-1})^{k-1}\Delta_k\Delta_{k-1}^{-1}.}
\end{equation}
In particular, $[Q_i,Q_j]=0$.
\end{lemma}
\begin{proof}
Let $A$ be the square block defining $\Delta_k$.
The quantum inverse-matrix formula, in our normalization, gives
\[
 (A^{-1})_{1k}=(-q)^{k-1}
 \qdet(A^{k1})\,(\qdet A)^{-1}.
\]
By the quasideterminant inverse formula,
$(A^{-1})_{1k}=|A|_{k,1}^{-1}$ and
$\qdet(A^{k1})=\Delta_{k-1}$. Taking inverses and using
commutativity of $\Delta_{k-1}$ and $\Delta_k$ gives
\eqref{eq:ratio}; see \cite{EtingofRetakh1998}.
\end{proof}

The mixed RLL relations imply
\begin{equation}\label{eq:weightC}
 K_sC_{ab}K_s^{-1}
 =q^{\delta_{s,n-b+1}-\delta_{s,n-a+1}}C_{ab}.
\end{equation}
Consequently, by \eqref{eq:ratio},
\begin{equation}\label{eq:weightQ}
 K_sQ_kK_s^{-1}
 =q^{\delta_{sk}-\delta_{s,n-k+1}}Q_k.
\end{equation}

\begin{proof}[Proof of Theorem~\ref{thm:main}]
We now return to the main assertion. Since the relevant minors are
of flag type, \eqref{eq:IQ} gives
$\Phi(I_i)=K_{n-i+1}Q_i$. From \eqref{eq:weightQ},
\begin{equation}\label{eq:beta}
 K_{n-i+1}Q_j=q^{\beta_{ij}}Q_jK_{n-i+1},\qquad
 \beta_{ij}=\delta_{i,n-j+1}-\delta_{ij}.
\end{equation}
These exponents are symmetric, $\beta_{ij}=\beta_{ji}$.
Together with commutativity of the $Q_i$ and of the $K_s$, this gives
\begin{align*}
 \Phi(I_i)\Phi(I_j)
 &=q^{-\beta_{ji}}
 K_{n-i+1}K_{n-j+1}Q_iQ_j\\
 &=q^{-\beta_{ij}}
 K_{n-j+1}K_{n-i+1}Q_jQ_i
 =\Phi(I_j)\Phi(I_i).
\end{align*}
Injectivity of $\Phi$ completes the proof.
\end{proof}

\section{A Gelfand--Tsetlin-type construction}

Lemma~\ref{lem:centers} provides a chain of nested algebras with
distinguished central elements:
\[
 \Delta_k\in Z(\mathcal B_k),\qquad
 \mathcal B_1\subset\cdots\subset\mathcal B_n.
\]
This resembles the Gelfand--Tsetlin subalgebra constructed from the
centers of nested $U_q(\mathfrak{gl}_k)$. In our setting, however, the
$\mathcal B_k$ are generated by entries of square RTT blocks. It would
be interesting to investigate commutative families arising from more
general chains of nested blocks.

The nested minors $\Delta_k$ commute, and their ratios $Q_k$
describe diagonal coordinates in the noncommutative Gauss
factorization; see \eqref{eq:Fk} and \eqref{eq:ratio}. The $Q_k$ should
not be regarded as central elements of the entire Borel subalgebra:
whenever their Cartan weights are nonzero, \eqref{eq:weightQ}
precludes centrality. The analogy is thus with the \emph{method} of
using centers of nested algebras, not with a literal identification of
our construction with the Gelfand--Tsetlin subalgebra.

A simple corollary yields commuting elements directly in the quantum
group, without inverting minors. Put
\[
 A_k=\prod_{r=n-k+1}^{n}K_r,
 \qquad \mathcal G_k=A_k\Delta_k.
\]
Then $[\mathcal G_i,\mathcal G_j]=0$. Indeed, the weights of
$\Delta_j$ imply
\[
 A_i\Delta_j A_i^{-1}=q^{\mu_{ij}}\Delta_j,
 \qquad \mu_{ij}=\max(0,i+j-n)-\min(i,j)
 =\mu_{ji},
\]
and the diagonal factors cancel exactly as in the proof of
Theorem~\ref{thm:main}.

\section{Quantum cluster coordinates}\label{sec:cluster}

The nested northeastern minors $\Delta_1,\ldots,\Delta_n$ are
natural quantum flag minors. In the theory of quantum cluster
algebras, such minors appear in cluster charts for quantum matrix and
flag algebras; see \cite{BZ2005,GoodearlYakimov2017}. Our proof does
not require choosing a particular cluster seed, but admits an
interpretation in terms of quantum tori.

Suppose that a chosen quantum cluster chart contains
$\Delta_1,\ldots,\Delta_n$ and allows them to be inverted. In its
quantum torus, the elements
\begin{equation}\label{eq:clusterMonomial}
 Q_k=(-q^{-1})^{k-1}\Delta_k\Delta_{k-1}^{-1}
\end{equation}
are Laurent monomials. By Lemma~\ref{lem:centers}, the restriction
of the $q$-commutation matrix to the $\Delta_k$ is zero. Thus the
ratios \eqref{eq:clusterMonomial} also form a commuting sector of the
torus. Upon adjoining the Cartan elements $K_s$, commutativity of the
images of the $I_k$ is ensured by the symmetry of the exponents
\eqref{eq:beta}.

The flag family can therefore be viewed as a system of
\emph{monomial coordinates of a commuting subtorus} in an extended
quantum cluster chart. However, the $Q_k$ need not themselves be
individual cluster $A$-variables or $X$-coordinates. Constructing an
explicit seed containing all $\Delta_k$ simultaneously and studying
its behavior under the passage from RTT to RE remain separate
questions.

\section{Relation to Skrypnyk's construction}\label{sec:skrypnyk}

Our construction has points of contact with the $k$-chop method
developed by T.~Skrypnyk for quantizing the full Toda system
\cite{Skrypnyk2000}. Unlike the $q$-deformed algebra considered here,
Skrypnyk works with the ordinary universal enveloping algebra
$U(\mathfrak{gl}_n)$ and its localizations.

It is natural to ask whether the family \eqref{eq:main} can occur as
part of a compatible $q$-deformation of Skrypnyk's full spectral
family.

\section{Conclusion and open problems}

The commutativity of the northeastern quasideterminants of the RE
matrix follows from three properties: covariance under triangular
transformations, centrality of quantum determinants of nested RTT
blocks, and symmetry of the Cartan weight exponents. No spectral
parameter is needed for the argument.

Several directions merit further investigation:
\begin{enumerate}
\item \emph{Cluster realization.} Construct an explicit quantum
cluster seed for the minors $\Delta_k$ and describe the resulting
commutative subtorus directly in the RE algebra.
\item \emph{Comparison with the $k$-chop method.} Compare the
$q$-deformed minor family with Skrypnyk's polynomial and rational
families, and determine whether a compatible specialization and a
natural extension of the subalgebra exist.
\item \emph{A spectral parameter.} Determine whether a commuting
spectral-parameter-dependent family exists in arbitrary rank.
\item \emph{The size of the commutative subalgebra.} Establish
algebraic independence of the chosen generators, maximality of the
resulting subalgebra in a suitable localization, and its properties
in representations.
\end{enumerate}

\section*{Acknowledgments}
The conjecture proved in this paper was formulated during the author's
stay at the Institut des Hautes \'Etudes Scientifiques (IHES) in 2024.
The author is grateful to IHES for its exceptionally stimulating and
productive research environment and, in particular, to Dmitry Gurevich
for inspiring discussions.

This work was carried out with the support of Russian Science
Foundation grant No.~26-11-00379, which funds the author's research
at P. G. Demidov Yaroslavl State University.

\section*{Declaration on the use of generative AI}
OpenAI's ChatGPT was used to assist with exploratory mathematical reasoning
and symbolic computations, language editing, English translation, and
preparation of the LaTeX manuscript. AI-assisted suggestions and computations
are not independent verification of the mathematical results. Responsibility
for the mathematical statements, references, and final manuscript rests with
the author.


\begin{thebibliography}{99}
\bibitem{Talalaev2024}
D.~V.~Talalaev,
\emph{RE-algebras, quasi-determinants and the full Toda system},
Partial Differential Equations in Applied Mathematics
\textbf{11} (2024), 100898.
\url{https://arxiv.org/abs/2406.07434}.

\bibitem{FRT1990}
L.~D.~Faddeev, N.~Yu.~Reshetikhin, L.~A.~Takhtajan,
\emph{Quantization of Lie groups and Lie algebras},
Leningrad Math. J. \textbf{1} (1990), 193--225.

\bibitem{JordanWhite2020}
D.~Jordan, N.~White,
\emph{The center of the reflection equation algebra via quantum minors},
J. Algebra \textbf{542} (2020), 308--342.
\url{https://doi.org/10.1016/j.jalgebra.2019.08.038}.

\bibitem{GelfandRetakh1991}
I.~M.~Gelfand, V.~S.~Retakh,
\emph{Determinants of matrices over noncommutative rings},
Funct. Anal. Appl. \textbf{25} (1991), 91--102.

\bibitem{GelfandRetakh1997}
I.~Gelfand and V.~Retakh,
\emph{Quasideterminants, I},
Selecta Math. (N.S.) \textbf{3} (1997), 517--546.
\url{https://doi.org/10.1007/s000290050019}.

\bibitem{EtingofRetakh1998}
P.~Etingof, V.~Retakh,
\emph{Quantum determinants and quasideterminants},
\url{https://arxiv.org/abs/math/9808065}.

\bibitem{Isaev2004}
A.~P.~Isaev, \emph{Quantum groups and Yang--Baxter equations},
MPIM Preprint Series, no.~132 (2004), \S~3.4.
\url{https://archive.mpim-bonn.mpg.de/id/eprint/3447/}.

\bibitem{BZ2005}
A.~Berenstein, A.~Zelevinsky,
\emph{Quantum cluster algebras},
Adv. Math. \textbf{195} (2005), 405--455.

\bibitem{GoodearlYakimov2017}
K.~R.~Goodearl, M.~T.~Yakimov,
\emph{Quantum cluster algebra structures on quantum nilpotent algebras},
Mem. Amer. Math. Soc. \textbf{247} (2017), no.~1169.
\url{https://arxiv.org/abs/1309.7869}.

\bibitem{Skrypnyk2000}
T.~Skrypnyk,
\emph{Commutative integrability of the full quantum Toda system},
Methods Funct. Anal. Topology \textbf{6} (2000), no.~4, 70--84.
\url{https://mfat.imath.kiev.ua/article/?id=143}.

\bibitem{BerensteinRetakh2005}
A.~Berenstein, V.~Retakh,
\emph{Noncommutative double Bruhat cells and their factorizations},
Int. Math. Res. Not. (2005), no.~8, 477--516.
\url{https://arxiv.org/abs/math/0407010}.
\end{thebibliography}
\end{document}